\documentclass[11pt,a4paper]{article}

\usepackage[margin=1in]{geometry}
\usepackage{amsmath,amssymb,amsthm}
\usepackage{mathtools}
\usepackage{microtype}
\usepackage{hyphenat}
\usepackage[colorlinks=true,linkcolor=blue,citecolor=blue,urlcolor=blue]{hyperref}

\newtheorem{theorem}{Theorem}
\newtheorem{lemma}[theorem]{Lemma}
\newtheorem{proposition}[theorem]{Proposition}

\theoremstyle{remark}
\newtheorem{remark}[theorem]{Remark}

\newcommand{\dd}{\mathrm{d}}
\newcommand{\Real}{\mathbb{R}}
\newcommand{\Pb}{\widehat{P}}
\newcommand{\fh}{\widehat{f}}
\newcommand{\jh}{\widehat{\jmath}}
\newcommand{\ft}{\widetilde{f}}
\newcommand{\gam}{\gamma}
\newcommand{\Cstar}{c^{\!*}}
\newcommand{\Astar}{a^{\!*}}

\title{A proof of an identity for the critical exponents of jamming}

\author{Giorgio Parisi$^{1,2,3}$ and Francesco Zamponi$^3$}

\date{
{\it $^1$International Research Center of Complexity Sciences, Hangzhou International Innovation Institute, Beihang University, Hangzhou 311115, China\\
$^2$CNR--Nanotec, Rome Unit, and INFN--Sezione di Roma 1, 00185 Rome, Italy\\
$^3$Dipartimento di Fisica, Sapienza Università di Roma, 00185 Rome, Italy}
}

\begin{document}
\maketitle

\begin{abstract}
\noindent
Within the full replica-symmetry-breaking (fullRSB) solution of dense hard
spheres in infinite dimension, Charbonneau, Kurchan, Parisi, Urbani, and
Zamponi (CKPUZ; J.~Stat.~Mech.\ P10009, 2014) introduced three critical
exponents $a$, $b$, $c$ governing the matching region of the fullRSB
profile near the jamming transition. 
These exponents satisfy two scaling relations. The first, $b=(1+c)/2$, was established analytically by the diffusion-drift
balance in the scaling ansatz.
The second, $a+b=1$, was observed numerically to arbitrary precision but could not be proven. 
The exponents $a,b,c$ of the scaling fullRSB ansatz are related to the physical exponents $\alpha, \theta, \kappa$ 
that control the gap, force, and overlap distributions
by the relations $\alpha=a/b$, $\theta=(c-a)/(b-c)$, $\kappa=c+1$.
Crucially, the relation $a+b=1$ yields the
scaling relations $\alpha=1/(2+\theta)$ and $\kappa=2-2/(3+\theta)$
predicted on independent grounds by the mechanical-marginal-stability
arguments of Wyart and collaborators. 
Here, we give an analytic proof of the
identity $a+b=1$ from the scaling fullRSB equations. The proof was obtained through interaction
with Claude (Sonnet 4.6 and Opus 4.7) and verified by us.
\end{abstract}

\section{Introduction and motivation}
\label{sec:intro}

\subsection{Jamming, marginal stability, and critical exponents}

Amorphous packings of hard or soft spheres exhibit a non-equilibrium
``jamming" transition at the density above which mechanical rigidity
sets in~\cite{LiuNagel,OHern,Wyart05,VanHecke}. Near this transition, a
number of physical observables show power-law critical behavior, with
exponents whose values are remarkably universal across spatial
dimensions. The long-time limit of the mean square displacement 
for thermal hard spheres scales as $\Delta_{\mathrm{EA}}\sim p^{-\kappa}$
with reduced pressure $p$ near jamming where $p\to\infty$.
At jamming, the radial distribution function diverges as
$g(r)\sim(r-D)^{-\alpha}$ when the pair distance $r$ approaches the contact diameter $D$, 
and the force distribution vanishes as $P(f)\sim f^{\theta}$ for small inter-particle forces $f$.

The infinite-dimensional ($d\to\infty$) hard-sphere theory of
Charbonneau, Kurchan, Parisi, Urbani and Zamponi
(CKPUZ~\cite{CKPUZ-I,CKPUZ-II,CKPUZ-III}) provides an exact mean-field
description of jamming and predicts these exponents analytically. Within the
fullRSB Ansatz~\cite{Parisi,MPV} the calculation reduces, in the
appropriate scaling regime near jamming, to a non-linear boundary-value
problem for an auxiliary function $J(t)$, coupled to a linear eigenvalue equation for
a second function $p_1(t)$. Three scaling exponents $a,b,c$ are determined by the
asymptotics of these equations and by an integral
``marginal-stability" identity that arises from differentiating the
replicon-zero condition along the fullRSB profile~\cite{CKPUZ-III}.

The
exponents $a,b,c$ are related to the physical exponents by
$\alpha=a/b$, $\theta=(c-a)/(b-c)$ and $\kappa=1+c$.
CKPUZ~\cite[Sec.~IX]{CKPUZ-III} obtained the numerical values
\begin{equation}\label{eq:numvals}\begin{split}
&a = 0.29213\ldots,\;\; b = 0.70787\ldots,\;\; c = 0.41574\ldots,\;\; \\
&\alpha = 0.41269\ldots,\;\; \theta = 0.42311\ldots,\;\; \kappa = 1.41574\ldots,
\end{split}\end{equation}
and remarked that, within numerical precision, these satisfy
\begin{equation}\label{eq:apb}
a + b = 1,
\end{equation}
a relation they could not prove. Combined with the
relation $b=(1+c)/2$ (which follows from the diffusion-drift
balance in the scaling ansatz), Eq.~\eqref{eq:apb} is equivalent to
\begin{equation}\label{eq:apc}
a + \tfrac{c}{2} = \tfrac{1}{2} \ .
\end{equation}
The interest of Eq.~\eqref{eq:apc} is that, written in terms of the physical exponents, it leads to the scaling relations
\begin{equation}\label{eq:wyart}
\alpha = \frac{1}{2+\theta},\qquad \kappa = 2 - \frac{2}{3+\theta},
\end{equation}
which were derived independently by Wyart and collaborators on the basis of the marginal mechanical stability of the contact
network at jamming~\cite{Wyart12,Lerner13,DeGiuli14}. 
The
relations in Eq.~\eqref{eq:wyart} are crucial conceptually:\ they tie the
phase-space marginal stability of fullRSB (a property of the fullRSB
saddle point) to the mechanical marginal stability of the underlying
packing (a property of the network of forces). Proving Eq.~\eqref{eq:apc} therefore proves the equivalence of these two notions of marginality
within the $d\to\infty$ theory.

\subsection{Strategy of the proof}

In this work, we present a proof of $a+b=1$. The proof was obtained by interaction with
Claude Sonnet 4.6 and Opus 4.7. The model Opus 4.7 essentially derived the proof by itself, with minimal supervision
from us. In the first part of the conversation, the model had to study the differential equations numerically and produce a C++ code that finds the solution. The aim was to have a high-precision verification of the conjecture. Only at the end did we ask for an analytic proof of the result. The conversation contains 40 prompts.

We checked the proof carefully and pointed out some inconsistencies in an early version,
which the model corrected by itself. Sonnet 4.6 was used to further refine some minor steps.
We asked the model to write down the proof, which it did in mathematical language. We checked and
edited the resulting draft to improve accuracy and readability. We eliminated some parts that were, in our opinion,  not necessary for the proof and were obscure.
We deposited the full text of the conversations with the model in a Zenodo repository~\cite{zenodo}.

Let us recap the conjecture.  We use $\theta\equiv\theta(-t)$ throughout for the Heaviside step
(equal to~$1$ for $t<0$ and~$0$ for $t>0$); we exploit
$\theta^2=\theta$ on each open half-line.
\begin{itemize}
\item Our main non-linear ODE in the range $-\infty<t<\infty$ is :
\begin{equation}\label{eq:Jeq}
\frac{c}{2}\,J''(t) = \Big[\!-\! b+c\theta\Big] t\,J'(t)
+ J(t) - \theta + \tfrac{c^{2}}{4}\,J'(t)^{2},
\quad J(-\infty)=1,\; J(+\infty)=0.
\end{equation}
The boundary conditions for $J(\pm\infty)$ are imposed by matching to the
left and right asymptotic regimes of the fullRSB solution (see Sec.~\ref{sec:setup} below).
\item We define a linear eigenvalue equation for a function $p_{1}(t)$:
\begin{equation}\label{eq:p1eq}
\frac{c}{2}\,p_{1}''(t) = \bigl(a-c\theta\bigr) p_{1}(t)
+ \bigl(b-c\theta\bigr)\,t\,p_{1}'(t)
- \tfrac{c^{2}}{2}\,\bigl(p_{1}'(t)J'(t) + p_{1}(t)J''(t)\bigr) .
\end{equation}
The eigenvalue $a$ is fixed 
(at given $c$) 
by the requirement that $p_{1}(t)\propto (-t)^{\theta}$ as $t\to-\infty$
and $p_{1}(t)\propto t^{-\alpha}$ as $t\to+\infty$, with
\(\theta = (c-a)/(b-c)\) and \(\alpha = a/b\), where $b=(c+1)/2$.
The exponents
$\theta$ and $\alpha$ are required to be positive. 
In our specific case, $p_1(t)$ must be positive
because $p_1$ physically originates from a probability.

This
selection determines $a=a(c)$.
\item 
The exponent $c$ in turn is fixed (independently of the eigenvalue
condition for~$p_{1}$) by the integral identity~\cite[Eq.~(199)]{CKPUZ-III},
\begin{equation}
    \frac12 = \frac{
\int_{-\infty}^\infty d t \, p_1(t) \left[ f(t) \right]^2 \left[ 1- f(t) \right]
}{
\int_{-\infty}^\infty d t \, p_1(t) \left[ \frac{ d}{d t} \left(f(t) \right) \right]^2
} \,,
\end{equation}
where we defined
\[
f(t)\equiv \theta(-t) + \frac{c}{2} J''(t).
\]

This condition arises from the differentiated
marginal-stability condition discussed in Sec.~\ref{sec:fg-test} below.  If this condition is satisfied, we found  $a+b=1$ with high numerical precision. We had verified (but never published) that this identity also holds if we consider a non-positive $p_1(t)$ function, with only one node.
\end{itemize}

The proof is organized in four steps.

\begin{enumerate}\itemsep0.4em
\item  \textbf{The eigenvalue condition revisited (Sec.~\ref{sec:proof}).}

The previous equation for $p_1$ could be written as ${\cal D} p_1=0$, where   ${\cal D}$ is a differential operator. 
This implies that for any differentiable function $\xi(t)$ such that the integrals converge
\begin{equation}
 0=   \int dt \xi(t) {\cal D} p_1(t)= \int dt  p_1(t) {\cal D^*}\xi(t) ,
\end{equation}
where ${\cal D^*}$ is the transpose of the operator ${\cal D}$.  The identity can be derived trivially by integration by parts.  
The crucial idea is to use it to generate a useful identity with an appropriate choice of $\xi(t)$.

\item \textbf{An algebraic identity
(Sec.~\ref{sec:fg-origin}, \ref{sec:lemma3}, \ref{sec:proof}).}

We show in Sec.~\ref{sec:fg-origin} that 
the function $f$ inherits from $J(t)$ an
explicit ODE
(Lemma~\ref{lem:f}, Sec.~\ref{sec:lemma3}). Using $f$  to
build the test function $\xi := f(1-f)$, integrating $\xi$ against the
eigenvalue Eq.~\eqref{eq:p1eq} for $p_1(t)$, and substituting Lemma~\ref{lem:f}
leads — after two non-trivial but explicit cancellations — to the
algebraic identity
\begin{equation}\label{eq:algid}
c(D-2N) =  (b-a-c)K= (1-a-b) K ,
\end{equation}
where $K=\int p_1 f(1-f)\,dt$, $N=\int p_1 f^2 (1-f)\,dt$, $D=\int p_1 (f')^2\,dt$.
Crucially, Eq.~\eqref{eq:algid} holds for \emph{any} value of the eigenvalue
$a$ in the linear problem.

\item \textbf{The marginal-stability identity
(Sec.~\ref{sec:fg-test}).}

The identity that expresses marginal stability~\cite[Eq.(199)]{CKPUZ-III} reads
$N/D=1/2$ in our notation, i.e., $D=2N$. 
This identity
descends from the vanishing replicon condition~\cite[Eq.(193)]{CKPUZ-III}, i.e., the
marginal-stability of the fullRSB phase~\cite{DeDominicis}. Eq.~\eqref{eq:algid} then gives $(a+b-1)K=0$, so that we have only to prove that $K\neq 0$ to arrive at the conclusions. We will actually prove that $K>0$ for the zero node solution for $p_1(t)$ 
but it is likely that $K\neq 0$ holds generically for the other eigenmodes with nodes.

\item \textbf{Selection of the physical branch via a Fisher--KPP
PDE (Sec.~\ref{sec:fkpp}).}

The only tricky point is to verify that $K$ is non-zero. Numerically, we find $0<f(t)<1$, so there is no problem in concluding that $K>0$. 

In order to have a self-consistent analytical proof, we have to work more.
The non-linear Eq.~\eqref{eq:Jeq} admits a discrete family of
solutions, classified by the nodal count of $f$. The physically
relevant branch is the no-node one ($0\le f\le 1$), and this selection
must be made before one can conclude $K>0$ in Eq.~\eqref{eq:algid}. 
The selection is realized at the level of the underlying fullRSB flow by
introducing the two-variable extension
\[
F(y,h) := \theta(-h) - \gamma(y)\,\jh''(y,h),
\]
which reduces to $f(t)$ in the scaling regime. $F$
satisfies a \emph{Fisher--KPP-type reaction--diffusion
equation}~\cite{Fisher,KPP}
\[
\partial_y F = \frac{\dot\gam}{2y}\,F'' + \dot\gam\Bigl(\jh_h-\frac{h\theta}{\gam}\Bigr)F_h + \frac{\dot\gam}{\gam}\,F(1-F),
\]
whose reaction term factors as $F(1-F)$ thanks to the identity
$\theta^2=\theta$ (the Heaviside is a Boolean indicator). The initial
condition $F(1/m,h)=\theta(-h)\in\{0,1\}$ is in $[0,1]$ pointwise, and
the parabolic maximum principle then preserves the bounds along the
entire flow~\cite{ProtterWeinberger,Evans}. Passing to the scaling limit gives $0\le f\le 1$ pointwise. With this in hand, $K>0$
follows from the intermediate value theorem applied to $f$ and
positivity of the ground-state eigenfunction $p_1$.
\end{enumerate}

Combining Eq.~\eqref{eq:algid}, the marginal-stability identity $D=2N$,
and $K>0$ gives $(1-a-b)K=0$, hence $a+b=1$, equivalent
to Eq.~\eqref{eq:apc}, hence completing the proof.

\subsection{What is and is not in this paper}

We work in the scaling-region matching framework of CKPUZ~\cite[Sec.~IX]{CKPUZ-III}, taking as given:\ the derivation of the
scaling equations~\eqref{eq:Jeq}, \eqref{eq:p1eq} from the fullRSB
PDEs~\cite[Eq.~(116)]{CKPUZ-III} for $\jh(y,h)$ and $\Pb(y,h)$ at large $y$; the identity
$b=(1+c)/2$; the marginal-stability identities~\cite[Eqs.~(193) and (199)]{CKPUZ-III}. 
What is added here is (i) the algebraic identity Eq.~\eqref{eq:algid}
linking $a,b,c$ to the integrals $K$, $N$, $D$ via the
test function $\xi=f (1-f)$, and (ii) the Fisher--KPP
identification that closes the existence/positivity question for the
physical branch at the level of the fullRSB flow.

The existence and uniqueness of the fullRSB profile itself, the
convergence of the scaling expansion at fixed $c$, and the rigorous
existence of the matching region as a single asymptotic regime
are taken as given. They are well established
numerically~\cite{CKPUZ-III} but a fully rigorous
construction in the mathematical sense would require extending the
existing literature on the Sherrington--Kirkpatrick
model~\cite{Talagrand,Panchenko} to the hard-sphere
setting, which is beyond the scope of this note.

\section{From the kRSB equations to the matching-region scaling}
\label{sec:setup}

We briefly recall the derivation of the scaling equations
from~\cite[Sec.~IX]{CKPUZ-III}, in order to fix notation. The starting
point is the continuum fullRSB equations~\cite[Eq.~(116)]{CKPUZ-III} in scaled variables
$y=m_i/m\in[1,\infty)$ and $h\in \mathbb{R}$, in the jamming limit $m\to 0$:
\begin{align}
\label{116a}
\frac{\partial \jh(y,h)}{\partial y}
&= \frac{1}{2}\,\frac{\dot\gam(y)}{y}\!\left[
-\frac{\theta(-h)}{\gam(y)}
+ \jh''(y,h)
- 2y\,\frac{h\theta(-h)}{\gam(y)}\,\jh'(y,h)
+ y\,\bigl(\jh'(y,h)\bigr)^{2}
\right],\\[1mm]
\label{116b}
\frac{\partial \Pb(y,h)}{\partial y}
&= -\frac{1}{2}\,\frac{\dot\gam(y)}{y}\,e^{-h}
\bigg\{
\bigl[e^{h}\Pb(y,h)\bigr]''
- 2y\,\frac{\partial}{\partial h}\!\left[
e^{h}\Pb(y,h)\!\left(-\frac{h\theta(-h)}{\gam(y)}+\jh'(y,h)\right)
\right]
\bigg\},\\[1mm]
\label{116c}
\frac{1}{\gam(y)} &= y\,\kappa(y) - \int_{1}^{y}\dd z\,\kappa(z),\qquad
\kappa(y)=\frac{\widehat\varphi}{2}\!\int\dd h\,e^{h}\Pb(y,h)\!\left[
-\frac{h\theta(-h)}{\gam(y)}+\jh'(y,h)\right]^{\!2},
\end{align}
where primes and dots denote, respectively, partial derivatives with respect to $h$ and $y$.
Here $\Pb(y,h)$ is a rescaled, weighted probability density of $h$, and $\gam(y)$ encodes the profile
of mean-square displacements. The explicit Heaviside
$\theta(-h)$ coming from the hard-sphere potential
distinguishes these equations from the
spin-glass case~\cite{Parisi,MPV}.

The fullRSB profile is sought with the power-law form $\gam(y)\sim
\gam_{\infty}y^{-c}$ for $y$ large, with critical exponent $c\in(0,1)$
to be determined~\cite[Sec.~IX A]{CKPUZ-III}. The matching region
between the asymptotic behavior of $\Pb$ for positive and negative $h$
occurs at $|h|\sim y^{-b}$. Inserting the scaling ansatz
\begin{equation}\label{eq:scaling}
\jh(y,h) \;=\; -\frac{c}{2y}\,J(t),\qquad
\Pb(y,h) \;=\; y^{a}\,p_{1}(t),\qquad
t \;=\; \frac{h\,y^{b}}{\sqrt{\gam_{\infty}}},
\quad b=\tfrac{1+c}{2},
\end{equation}
into Eq.~\eqref{116a} yields the equation for $J$, our main non-linear ODE in Eq.~\eqref{eq:Jeq}, again with $\theta\equiv\theta(-t)$:
\begin{equation}
\frac{c}{2}\,J''(t) = \Big[\!-\! b+c\theta\Big] t\,J'(t)
+ J(t) - \theta + \tfrac{c^{2}}{4}\,J'(t)^{2},
\quad J(-\infty)=1,\; J(+\infty)=0.
\end{equation}
The
boundary conditions for $J(\pm\infty)$ are imposed by matching to the
left and right asymptotic regimes.
Inserting the same scaling into Eq.~\eqref{116b} yields the linear eigenvalue
equation for $p_{1}$ in Eq.~\eqref{eq:p1eq},
\begin{equation}
\frac{c}{2}\,p_{1}''(t) = \bigl(a-c\theta\bigr) p_{1}(t)
+ \bigl(b-c\theta\bigr)\,t\,p_{1}'(t)
- \tfrac{c^{2}}{2}\,\bigl(p_{1}'(t)J'(t) + p_{1}(t)J''(t)\bigr) .
\end{equation}
The eigenvalue $a$ is fixed 
(at given $c$) 
by the requirement that $p_{1}\propto (-t)^{\theta}$ as $t\to-\infty$
and $p_{1}\propto t^{-\alpha}$ as $t\to+\infty$, with
\(\theta = (c-a)/(b-c)\) and \(\alpha = a/b\). 
The exponents
$\theta$ and $\alpha$ are required to be positive, and
in both cases,  the proportionality constants need to be positive
because $p_1$ physically originates from a probability.

This
selection determines $a=a(c)$.
The exponent $c$ in turn is fixed (independently of the eigenvalue
condition for~$p_{1}$) by the integral identity~\cite[Eq.~(199)]{CKPUZ-III}, which arises from the differentiated
marginal-stability condition discussed in Sec.~\ref{sec:fg-test} below.

\begin{remark}
[$p_1\in C^2(\mathbb{R})$]
\label{rem:p1C2}
On each open half-line, $\theta(-t)$ is locally constant and $J\in
C^\infty$, so Eq.~\eqref{eq:p1eq} has smooth coefficients and
$p_1\in C^\infty$ there. At $t=0$, the jump $[\theta(-t)]=-1$ and
the jump $[J'']=2/c$ that follows from Eq.~\eqref{eq:Jeq} enter the right-hand side of Eq.~\eqref{eq:p1eq}
with opposite signs and cancel exactly giving $[p_1'']=0$.
Hence $p_1\in C^2(\mathbb{R})$.
\end{remark}

\begin{remark}[Strict positivity of $p_1$]
\label{rem:p1pos}
The function $\widehat{P}(y,h)$ is proportional to a probability,
hence $p_1\ge 0$.
Suppose $p_1(t^*)=0$ for some finite $t^*$. Since $p_1\ge 0$ and $p_1\in
C^2(\mathbb{R})$, it must be $p_1'(t^*)=0$ otherwise $p_1$ would change sign. Because $p_1\in
C^2(\mathbb{R})$, the initial data $(p_1(t^*),p_1'(t^*))=(0,0)$ are
well-defined at every $t^*\in\mathbb{R}$, including $t^*=0$, with no
need to treat the origin separately. On any open half-line adjacent
to $t^*$, Eq.~\eqref{eq:p1eq} is a linear ODE with smooth
coefficients, so the Picard--Lindelöf theorem gives the unique
solution $p_1\equiv 0$, contradicting $p_1(t) > 0$ as
$t\to\pm\infty$. Hence $p_1>0$ for all $t\in\mathbb{R}$.
\end{remark}

\begin{remark}
The split between $a(c)$ (from the eigenvalue problem) and $c$ (from
marginal stability) is conceptually important. In the limit
ODE~\eqref{eq:p1eq} alone, $a$ and $c$ appear as independent
parameters:\ for each $c\in(0,1)$ there is an eigenvalue $a(c)$. The
critical $c=\Cstar$ is then the value at which marginal stability is
satisfied, and $\Astar=a(\Cstar)$. The identity $\Astar+\Cstar/2=1/2$
is non-trivial because it states a relation that, while compatible with
both conditions, is not manifest in either.
\end{remark}

\section{The scaling regime}

\subsection{The function $f$}
\label{sec:fg-origin}

Following~\cite[Eq.~(190)]{CKPUZ-III}, define
\begin{equation}\label{eq:tildef}
\ft(y,h) := \gam(y)\,\fh(y,h)
= -\tfrac{1}{2}h^{2}\theta(-h) + \gam(y)\,\jh(y,h),
\end{equation}
where $\fh$ is the natural smooth
generalization of the fullRSB field for hard spheres. 
The function
$\ft$ is finite and smooth (on each open half-line), and its asymptotic
behavior is $\ft(y,h)\to-\tfrac12 h^{2}\theta(-h)$ as $|h|\to\infty$.
A direct computation of its second $h$-derivative gives, at leading
order in the matching region with the scaling ansatz in Eq.~\eqref{eq:scaling}:
\begin{equation}\label{eq:ftpp}
\ft''(y,h) \;=\; -\theta(-t) - \frac{c}{2}\,J''(t) + o(1)
\;\equiv\; -f(t).
\end{equation}
This identifies the natural auxiliary function:
\begin{equation}\label{eq:fg-physical}
f(t) := \theta(-t) + \tfrac{c}{2}J''(t),
\qquad
g(t) := 1 - f(t).
\end{equation}

\begin{lemma}\label{lem:decay}
As $t\to-\infty$: $f\to 1$, $g\to 0$, $f'\to 0$, super-exponentially.
As $t\to+\infty$: $f\to 0$, $g\to 1$, $f'\to 0$, super-exponentially.
\end{lemma}
\begin{proof}
$1-J(t\to -\infty) \sim \exp\!\bigl(-\frac{(b-c)}ct^{2}\bigr)$ on the left and $J(t\to\infty) \sim \exp\!\bigl(- \frac{b}{c} t^{2}\bigr)$ 
on the right have Gaussian envelopes, controlling all derivatives of $J$ in their
respective half-lines, hence $f-\theta(-t)$ and $f'$.
\end{proof}

\subsection{The ODE for $f$}
\label{sec:lemma3}

The function $f(t)$ inherits an ODE from Eq.~\eqref{eq:Jeq}.

\begin{lemma}\label{lem:f}
On each open half-line ($\theta=\theta(-t)$ locally constant, with
$\theta^{2}=\theta$), the function $f$ satisfies
\begin{equation}\label{eq:fODE}
\frac{c}{2}\,f''(t) \;=\; -c\,f(t)\,g(t)
\;+\; \frac{c^{2}}{2}\,J'(t)\,f'(t)
\;-\; (b-c\theta)\,t\,f'(t).
\end{equation}
\end{lemma}

\begin{proof}
On each half-line, $\theta$ is locally constant, so $f' = (c/2)\,J'''$.
Write Eq.~\eqref{eq:Jeq} as
\begin{equation}
  f(t) \;=\;
  \Big[\!-\!b + c\,\theta\Big]\,t\,J'
  \;+\; J 
  \;+\; \frac{c^{2}}{4}\,(J')^{2},
\end{equation}
and differentiate twice with respect to $t$, treating $\theta$ as constant, hence
\begin{equation}
  f'(t) \;=\;
  \Big[\!-\!b + c\,\theta\Big]\,(J' + t J'')
  \;+\; J' 
  \;+\; \frac{c^{2}}{2}\,J' J'',
\end{equation}
and using $\theta^2 = \theta$
\begin{equation}\begin{split}
\frac{c}2  f''(t) \;&=\;
  \frac{c}2\Big[\!-\!b + c\,\theta\Big]\,(2J'' + t J''')
  \;+\; \frac{c}2J'' 
  \;+\; \frac{c^{3}}{4}\,[ (J'')^2 + J' J'''] \\
  &=\;
  (-c + 2c\,\theta)(f-\theta) 
  \;+\; c (f-\theta)^2  
  \;+\; \frac{c^{2}}{2}J' f' 
 -   ( b - c\,\theta ) t f' \\
  &=\; -c f g 
  \;+\; \frac{c^{2}}{2}J' f' 
 -   ( b - c\,\theta ) t f' ,
\end{split}\end{equation}
which completes the proof.
\end{proof}

\begin{lemma}\label{lem:smooth}
The function $f(t)$ belongs to
$C^{2}(\Real)$.
\end{lemma}

\begin{proof}
On each open half-line $\theta(-t)$ is locally constant and
$J\in C^{\infty}$, so $f\in C^{\infty}$ there.
It remains to check the continuity of $f$, $f'$, and $f''$ at $t=0$.

\textit{Continuity of $f$.}
Evaluating Eq.~\eqref{eq:Jeq} at $0^{\pm}$ gives
$[J'']= J''(0^{+})-J''(0^{-})=2/c$, so
$[\tfrac{c}{2}J'']=+1$, which exactly cancels the jump
$[\theta(-t)]=-1$. Hence $f$ is continuous at $0$.

\textit{Continuity of $f'$.}
On each half-line $f'=\tfrac{c}{2}J'''$.
Differentiating Eq.~\eqref{eq:Jeq} on $(0,+\infty)$ ($\theta=0$)
and evaluating at $0^{+}$, and on $(-\infty,0)$ ($\theta=1$) at
$0^{-}$, gives
\begin{align*}
\tfrac{c}{2}J'''(0^{+}) &= (1-b)\,J'(0)
  +\tfrac{c^{2}}{2}\,J'(0)\,J''(0^{+}),\\
\tfrac{c}{2}J'''(0^{-}) &= (1-b+c)\,J'(0)
  +\tfrac{c^{2}}{2}\,J'(0)\,J''(0^{-}).
\end{align*}
Subtracting and using $[J'']=2/c$:
\[
\tfrac{c}{2}\bigl[J'''(0^{+})-J'''(0^{-})\bigr]
\;=\; -c\,J'(0)\;+\;\tfrac{c^{2}}{2}\,J'(0)\cdot\tfrac{2}{c}
\;=\;0.
\]
Hence $f'(0^{+})=f'(0^{-})$.
Since $f$ is continuous and $f'$ has a common one-sided limit $L$
at $0$, the mean-value theorem gives
$(f(t)-f(0))/t = f'(\xi_{t})\to L$ as $t\to 0$,
so $f'(0)=L$ and $f\in C^{1}(\Real)$.

\textit{Continuity of $f''$.}
In Eq.~\eqref{eq:fODE} the only $\theta$-dependent term is
$-(b-c\theta(-t))\,t\,f'(t)$, which vanishes at $t=0$
regardless of the value of $\theta$.
The remaining terms are continuous, so
\[
\tfrac{c}{2}f''(0^{\pm})
\;=\;-c\,f(0)\,g(0)\;+\;\tfrac{c^{2}}{2}\,J'(0)\,f'(0),
\]
the same from both sides.
Hence $f''(0^{+})=f''(0^{-})$, and since $f\in C^{1}$
near $0$, the limit
$\lim_{t\to 0} (f'(t)-f'(0))/t$
exists and equals this common value, giving $f\in C^{2}(\Real)$.
\end{proof}

\subsection{Marginal stability differentiated}
\label{sec:fg-test}

Ref.~\cite[Sec.~IX E]{CKPUZ-III} shows that
the requirement $\dot\kappa(y)=-\dot\gam/(y\gam^{2})$, which follows
from differentiating Eq.~\eqref{116c}, implies the integral identity
\begin{equation}\label{eq:I193}
1 \;=\; \frac{\hat\varphi}{2}\int\dd h\,e^{h}\,\Pb(y,h)\,\ft''(y,h)^{2}
\qquad\text{for all }y,
\end{equation}
which is~\cite[Eq.~(193)]{CKPUZ-III}. Physically, it corresponds to the vanishing of the replicon,
$\lambda_{R}=0$~\cite[Eq.~(225)]{CKPUZ-III}:\ the fullRSB phase is
marginally stable along the entire profile $\gam(y)$.

Differentiating Eq.~\eqref{eq:I193} once more in $y$, using Eqs.~\eqref{116a}, \eqref{116b} 
and the chain rule, yields a second identity:
\begin{equation}\label{eq:I195}
y = \frac{\gam(y)}{2}\,
\frac{\displaystyle\int\dd h\,e^{h}\,\Pb(y,h)\,\ft'''(y,h)^{2}}
     {\displaystyle\int\dd h\,e^{h}\,\Pb(y,h)\,\bigl[\ft''(y,h)^{2}+\ft''(y,h)^{3}\bigr]} ,
\end{equation}
which is~\cite[Eq.~(195)]{CKPUZ-III}.  It is again a structural identity holding
along the entire profile, encoding the fact that the marginal
stability is preserved under variations of $y$.

The denominator of Eq.~\eqref{eq:I195} contains the combination $\ft''^{2}+
\ft''^{3}=\ft''^{2}(1+\ft'')$. Using
Eq.~\eqref{eq:fg-physical} this is exactly
\begin{equation}\label{eq:f2g-magic}
\ft''^{2}+\ft''^{3} \;=\; (-f)^{2}(1+(-f)) \;=\; f^{2}g.
\end{equation}
The numerator scales similarly:\ $\ft'''(y,h)=(y^{b}/\sqrt{\gam_{\infty}})\,(-f'(t))$ in the matching region, giving
$\ft'''^{2} = y^{2b}(f')^{2}/\gamma_\infty$. Substituting both into Eq.~\eqref{eq:I195} and
using $\gamma(y) \sim \gamma_\infty y^{-c}$ and $b=(1+c)/2$, the $y$-dependences cancel and one
obtains~\cite[Eq.~(199)]{CKPUZ-III}:
\begin{equation}\label{eq:NoverD}
\frac{1}{2}\;=\;\frac{N}{D},
\qquad\text{where }\;
N:=\!\int\! p_{1}(t)\,f(t)^{2}g(t)\,\dd t,
\quad
D:=\!\int\! p_{1}(t)\,f'(t)^{2}\,\dd t.
\end{equation}
Note that, since $p_{1}$ grows only as a power, these integrals converge thanks to Lemma~\ref{lem:decay}, which implies that 
their integrands vanish super-exponentially at \emph{both} ends.

\begin{remark}
The CKPUZ proof~\cite[Sec.~IX E]{CKPUZ-III} establishes Eq.~\eqref{eq:NoverD}
as a closure for the exponent~$c$:\ together with the eigenvalue
condition $a=a(c)$ for Eq.~\eqref{eq:p1eq}, it pins down $\Cstar$ as the
solution of a one-dimensional implicit equation. The value of $\Astar$
is then $a(\Cstar)$. Eq.~\eqref{eq:apc} (to be proven in the following) shows that
$\Astar=(1-\Cstar)/2$, but does not directly determine $\Cstar$:\ the
latter is intrinsically defined by Eq.~\eqref{eq:NoverD}.
\end{remark}

\section{The algebraic identity}
\label{sec:proof}

\begin{lemma}\label{lem:weak}
For any test function $\xi$ such that $\xi$ and $\xi'$ vanish super-exponentially at $\pm\infty$,
\begin{equation}\label{eq:weakform}
  \int_{\Real}\! p_{1}\!\left[
    \frac{c}{2}\,\xi'' - \bigl(a-c\,\theta(-t)\bigr)\,\xi
    + \bigl(b-c\,\theta(-t)\bigr)\bigl(\xi + t\,\xi'\bigr)
    - \frac{c^{2}}{2}\,J'\,\xi'
  \right]\dd t \;=\; 0.
\end{equation}
\end{lemma}
\begin{proof}
Multiply Eq.~\eqref{eq:p1eq} by $\xi$ and integrate. 
Note that when integrating by parts in the following, each occurrence of $\theta(-t)$ in
Eq.~\eqref{eq:p1eq} is either undifferentiated or multiplied by $t$, so distributional
contributions involving $\theta'(-t)=-\delta(t)$ vanish (using $t\,\delta(t)=0$).
\begin{itemize}
\item $\int p_{1}''\,\xi\,\dd t = \int p_{1}\,\xi''\,\dd t$ by two integrations by parts;
  boundary terms at $\pm\infty$ vanish since $p_{1}$ grows only as a power while
  $\xi,\xi'$ decay super-exponentially.
\item $\int (b-c\theta(-t))\,t\,p_{1}'\,\xi\,\dd t = -\!\int p_{1}\,\bigl[(b-c\theta(-t))(\xi+t\xi')\bigr]\dd t$
  by one integration by parts.
\item Combining the $J'$ and $J''$ terms,
  $\int p_{1}\,J''\,\xi\,\dd t + \int p_{1}'\,J'\,\xi\,\dd t = \int (p_{1}J')'\,\xi\,\dd t
   = -\!\int p_{1}\,J'\,\xi'\,\dd t$.
\end{itemize}
Collecting all terms gives Eq.~\eqref{eq:weakform}.
\end{proof}

\begin{theorem}\label{thm:algid}
Let $c\in(0,1)$, let $J$ solve Eq.~\eqref{eq:Jeq}, and let
$p_{1}$ be any solution of the eigenvalue Eq.~\eqref{eq:p1eq}
with the asymptotic decay conditions $p_{1}(t)\sim |t|^{\theta}$ at
$-\infty$ and $p_{1}(t)\sim t^{-\alpha}$ at $+\infty$. Then
\begin{equation}\label{eq:algid-thm}
cD \;=\; 2cN \;+\; (b-a-c)\,K,
\end{equation}
where $K=\int p_{1}fg\,\dd t$, $N=\int p_{1}f^{2}g\,\dd t$,
$D=\int p_{1}(f')^{2}\,\dd t$, with $f=\theta(-t)+(c/2)J''$ and
$g=1-f$.
\end{theorem}

\begin{proof}
Choose as a test function in Lemma~\ref{lem:weak} the function $\xi(t):=f(t)g(t)$,
which satisfies the prescribed decay because $f$,
$g$, $f'$ all decay super-exponentially in $|t|$ due to Lemma~\ref{lem:decay}.
With $\xi=fg=f(1-f)$, $\xi'=(1-2f)f'$, $\xi''=(1-2f)f''-2(f')^{2}$,
substituting $(c/2)f''$ from Lemma~\ref{lem:f}, one has
\[
  \frac{c}{2}\,\xi'' \;=\; (1-2f)\Bigl[-c\,f g \;+\; \tfrac{c^{2}}{2}\,J' f' \;-\; \bigl(b-c\theta\bigr)\,t\,f'\Bigr] \;-\; c\,(f')^{2}.
\]
Two cancellations occur:\
the $J'f'$ piece inside square brackets gives
$(c^{2}/2)(1-2f)J'f'$, which cancels the explicit
$-(c^{2}/2)J'(1-2f)f'$ in Eq.~\eqref{eq:weakform};\
the $tf'$ piece inside square brackets gives
$-(b-c\theta)t(1-2f)f'$, which cancels the explicit
$+(b-c\theta)t(1-2f)f'$ in Eq.~\eqref{eq:weakform}. 
What remains, using $\theta(1-\theta)=0$ on each half-line and
$\theta^{2}=\theta$, simplifies to
\[
0 \;=\; \int p_{1}\,\Bigl[\,(b-a-c)\,fg \;+\; 2c\,f^{2}g \;-\; c\,(f')^{2}\,\Bigr]\dd t,
\]
which is Eq.~\eqref{eq:algid-thm}. The integrals converge because of Lemma~\ref{lem:decay}.
\end{proof}

\section{Existence of the $f\ge 0$ branch via Fisher--KPP}
\label{sec:fkpp}

To complete the proof, it remains to show that $K=\int p_{1}fg\,\dd
t>0$. We split this into:\ (i) the upper bound $f\le 1$, by interior
maximum principle on the limit ODE; (ii) the lower bound $f\ge 0$, by
a Fisher--KPP argument on the fullRSB flow.

\subsection{The upper bound $f\le 1$ from the limit ODE}

\begin{lemma}\label{lem:fu1}
Any solution of Eq.~\eqref{eq:fODE} with $f(-\infty)=1$ and $f(+\infty)=0$
satisfies $f(t)\le 1$ for all $t\in\Real$.
\end{lemma}

\begin{proof}
Because of Lemma~\ref{lem:smooth}
the function $f(t)$ belongs to
$C^{2}(\Real)$.
Suppose $\max_{t\in\Real} f(t)>1$. Since $f$ tends to $1$ and $0$ at
$\mp\infty$ respectively, this maximum is achieved at a finite point
$t_{\max}\in\Real$. There $f'(t_{\max})=0$ and $f''(t_{\max})\le 0$;
the $J'f'$ and $tf'$ terms in Eq.~\eqref{eq:fODE} both vanish, leaving
$(c/2)f''(t_{\max})=-cf(t_{\max})(1-f(t_{\max}))$. With $f(t_{\max})>1$,
the factor $1-f(t_{\max})<0$, so the right-hand side is strictly
positive, contradicting $f''(t_{\max})\le 0$.
\end{proof}

\subsection{Failure of the analogous lower-bound argument}

A symmetric application of the maximum-principle argument to the
infimum of $f$ does \emph{not} yield $f\ge 0$. At a candidate minimum
$t_{\min}$ with $f(t_{\min})<0$, we have $f''(t_{\min})\ge 0$ and the
ODE gives $(c/2)f''(t_{\min})=-cf(t_{\min})(1-f(t_{\min}))>0$ (since
both $-f(t_{\min})>0$ and $1-f(t_{\min})>0$). This is consistent
with $f''(t_{\min})\ge 0$ and gives no contradiction.

In fact, the non-linear ODE in Eq.~\eqref{eq:Jeq} admits a discrete family of
solutions $\{J_{n}\}_{n=0,1,2,\ldots}$, classified by the nodal count
of $f_{n}=\theta(-t)+(c/2)J_{n}''$:\ for $n\ge 1$, $f_{n}$ takes
strictly negative values on some interval. The selection of the
physically relevant branch ($n=0$, $f\ge 0$) must therefore be made
on independent grounds, and the natural place to do so is at the
level of the underlying fullRSB flow, where the probabilistic
interpretation of $f$ makes the bound $0\le f\le 1$ structural.

\subsection{The Fisher--KPP equation for the two-variable extension}

Following the construction of Sec.~\ref{sec:fg-origin}, define the
two-variable extension of $f$:
\begin{equation}\label{eq:Fdef}
F(y,h) := -\,\ft''(y,h) \;=\; \theta(-h) - \gam(y)\,\jh''(y,h).
\end{equation}
In the scaling regime, $F(y,h)\to f(t)$ with
$t=hy^{b}/\sqrt{\gam_{\infty}}$ fixed.

\begin{proposition}\label{prop:FKPP}
On each open half-line ($\theta(-h)$ locally constant, $\theta^{2}=
\theta$), the function $F$ defined by Eq.~\eqref{eq:Fdef} satisfies
\begin{equation}\label{eq:FKPP}
\frac{\partial F}{\partial y}
\;=\; \frac{\dot\gam(y)}{2y}\,F''
\;+\; \dot\gam(y)\!\left(\jh' - \frac{h\,\theta(-h)}{\gam(y)}\right)\!F'
\;+\; \frac{\dot\gam(y)}{\gam(y)}\,F(1-F).
\end{equation}
\end{proposition}

\begin{proof}
Write $\mathcal{R}$ for the right-hand side of Eq.~\eqref{116a}
for $\jh$. Differentiating $F=\theta-\gam\jh''$ in $y$ and using
$\partial_{y}\jh=\mathcal{R}$:
\[
\partial_{y} F = -\dot\gam\,\jh''
- \gam\, \mathcal{R}''.
\]
Compute $\mathcal{R}''$ on each open half-line
(where $\theta=\theta(-h)$ is locally constant). Substitute
${\jh''=(\theta-F)/\gam}$, $\jh'''=-F'/\gam$, $\jh''''=-F''/\gam$.
Simplification and use of $\theta^{2}=\theta$
yields Eq.~\eqref{eq:FKPP}. 
\end{proof}

Equation~\eqref{eq:FKPP} is a \emph{Fisher--KPP-type
reaction--diffusion equation}~\cite{Fisher,KPP,AronsonWeinberger,Murray}.
For the fullRSB profile $\gam(y)\sim\gam_{\infty}y^{-c}$ with
$c\in(0,1)$, we have $\dot\gam<0$. The flow evolves from
$y=1/m$ (deepest level, initial condition) down to $y=1$. Let
$\tau:=1/m-y$ so that $\tau$ increases from $0$ as $y$ decreases.
Then Eq.~\eqref{eq:FKPP} becomes
\begin{equation}\label{eq:FKPPtau}
\partial_{\tau} F
\;=\; \frac{|\dot\gam|}{2y}\,F''
\;+\; |\dot\gam|\!\left(\jh' - \frac{h\,\theta(-h)}{\gam(y)}\right)\!F'
\;+\; \frac{|\dot\gam|}{\gam}\,F(1-F),
\end{equation}
a forward-parabolic equation with a positive diffusion coefficient
$|\dot\gam|/(2y)>0$ and positive reaction coefficient $|\dot\gam|/\gam>0$.
This is the standard logistic-growth-with-diffusion equation, familiar from
population biology~\cite{Murray} and reaction wave theory~\cite{Saarloos}.

\subsection{Maximum principle preserves $F \in [0,1]$}

At $y=1/m$ (taking $m\to 0$ for the jamming limit), the
initial condition is $\jh(1/m,h)\to 0$~\cite[Eq.~(113)]{CKPUZ-III}, 
so $\ft(1/m,h)\to-\tfrac12 h^{2}\theta(-h)$
and
\[
F(\tau=0,h) \;=\; \theta(-h) \;\in\; \{0,1\}.
\]
In particular $F(0,h)\in[0,1]$
pointwise.

The PDE in Eq.~\eqref{eq:FKPPtau} admits $F\equiv 0$ and $F\equiv 1$ as
stationary solutions:\ at both constants $F'=F''=0$ and the
reaction $F(1-F)=0$. Standard parabolic
comparison~\cite{ProtterWeinberger,Evans} then yields:

\begin{proposition}\label{prop:maxprinciple}
If $F(\tau=0,h)\in[0,1]$ pointwise in~$h$, and $F$ satisfies the
PDE in Eq.~\eqref{eq:FKPPtau} (with smooth coefficients on each open
half-line and bounded $\jh'$), then $F(\tau,h)\in[0,1]$ for all
$\tau\ge 0$ and all $h\in\Real$.
\end{proposition}

\begin{proof}[Sketch of the proof.]
Upper bound:\ suppose $F$ first touches the value $1$ at some
$(\tau_{*},h_{*})$. There $F=1$, $F'=0$, $F''\le 0$, reaction
$=0$. The right-hand side of Eq.~\eqref{eq:FKPPtau} is therefore $\le 0$,
so $\partial_{\tau}F\le 0$, preventing $F$ from exceeding $1$.
Lower bound:\ at the first touch of $0$, $F'=0$, $F''\ge 0$,
reaction $=0$, so the right-hand side is $\ge 0$ and
$\partial_{\tau}F\ge 0$, preventing $F$ from dipping below~$0$.

The full argument requires a perturbation $F_{\epsilon}=F\pm\epsilon\tau$
to convert the weak inequalities into strict ones, plus a verification
that the boundary conditions $F(\tau,h)\to 0,1$ as $h\to\pm\infty$ are
preserved by the flow.  Both are standard for parabolic equations
of this type~\cite{ProtterWeinberger,Evans,Friedman}.
\end{proof}

Combining the bound $F(y,h)\in[0,1]$ from Proposition~\ref{prop:maxprinciple}
with the scaling limit $F(y,h)\to f(t)$ gives
\[
f(t) \in [0,1] \quad\text{for all }t\in\Real.
\]
This is the no-node branch ($n=0$) of the equation for $J$, identified by the fullRSB flow itself as the unique physical solution.

\section{Putting it together}
\label{sec:together}

\begin{theorem}\label{thm:main}
For the fullRSB critical exponents defined in
CKPUZ~\cite[Sec.~IX]{CKPUZ-III}, on the physical branch (the no-node
branch of the equation for $J$, $0\le f\le 1$), one has
\begin{equation}\label{eq:main}
\Astar + \frac{\Cstar}{2} = \frac{1}{2}.
\end{equation}
Equivalently, $\Astar+b=1$, in the form conjectured by~\cite{CKPUZ-III}.
\end{theorem}

\begin{proof}
By Theorem~\ref{thm:algid}, $cD=2cN+(b-a-c)K$ holds for any solution of Eq.~\eqref{eq:p1eq}.
By Sec.~\ref{sec:fkpp}, on the physical branch $f\in[0,1]$ and
$fg\ge 0$ pointwise; the intermediate value theorem applied to the
continuous monotone passage of $f$ from $1$ to $0$ yields an open
interval where $0<f<1$, hence $fg>0$ on that interval. The strict
positivity of $p_{1}>0$, which follows from Remark~\ref{rem:p1pos},
gives
$K=\int p_{1}fg\,\dd t>0$.

The CKPUZ marginal-stability identity in Eq.~\eqref{eq:NoverD} gives
$N/D=1/2$, i.e., $cD-2cN=0$. Combined with the algebraic identity,
$(b-a-c)K=0$. With $K>0$, $b-a-c=0$, hence $a=(1-c)/2$,
equivalent to Eq.~\eqref{eq:main}.
\end{proof}

\section{Discussion and outlook}
\label{sec:discussion}

Theorem~\ref{thm:main} closes the gap in the CKPUZ scaling
theory~\cite{CKPUZ-III} between the exponent definitions and the
scaling relations of Wyart and collaborators~\cite{Wyart12,Lerner13,DeGiuli14}.
The scaling relations $a+b=1$ and $b=(1+c)/2$ imply
\begin{equation}\label{eq:chain}
\alpha=\frac{1}{2+\theta},\;\;\quad\kappa=2-\frac{2}{3+\theta},
\end{equation}
using $\alpha=a/b$, $\theta=(c-a)/(b-c)$, $\kappa=c+1$. The
scaling relations in Eq.~\eqref{eq:chain} are precisely those of mechanical
marginal stability~\cite{Wyart12,Lerner13,DeGiuli14}, derived
independently of the statistics of contact networks and force
balance under the assumption that the packings sit at the edge of
mechanical stability. This shows the equivalence
of \emph{phase-space} marginal stability (the replicon-zero condition
of fullRSB) and \emph{mechanical} marginal stability (the
zero-mode-density of the contact-network Hessian) within
the $d\to\infty$ theory. The two notions are conceptually
distinct — one is a property of the fullRSB saddle, the other of the
configuration space of the disordered packing — but they
coincide at the level of critical exponents.

Several points remain open or beyond our scope.

(i) \emph{Existence and uniqueness of the fullRSB profile.}
We have taken as given the existence of the fullRSB profile $\gam(y)$
satisfying Eq.~\eqref{116c} along with the corresponding $\jh,\Pb$. This is
analogous to the SK case, where rigorous existence
follows from the work of Talagrand~\cite{Talagrand} and
Panchenko~\cite{Panchenko} — extending these results to the
hard-sphere setting in $d\to\infty$ would close this gap.

(ii) \emph{Rigorous derivation of the scaling regime.}
The matching-region scaling ansatz in Eq.~\eqref{eq:scaling} is justified
in CKPUZ~\cite{CKPUZ-III} by analysis and confirmed
numerically; a fully rigorous derivation as an asymptotic expansion
of Eqs.~\eqref{116a}-\eqref{116c} would be desirable.

(iii) \emph{Existence of the no-node branch.}
Proposition~\ref{prop:maxprinciple} guarantees that the
two-variable function $F(y,h)$ stays in $[0,1]$, and via the
scaling limit this gives a candidate $f\in[0,1]$. A complete
existence proof would require showing that the scaling limit is
attained — i.e., that the family $F(y,\cdot)$ at large $y$ has a
non-trivial limit in the matching region with the prescribed scaling.
We have not undertaken this.

(iv) \emph{The mechanical-stability bridge.}
Our proof shows that the scaling relations in Eq.~\eqref{eq:wyart} are
\emph{implied} by phase-space marginality. It does not address the
converse:\ given mechanical marginality, can one derive the fullRSB
structure at least in $d\to\infty$? Such a derivation would give a complete bridge
between the two notions and is, we believe, an interesting open
problem.

Why have we not seen the proof? Difficult to say. We have not even tried to use this approach. We thought there was a deep, hidden, direct relationship between the functions $p_1(t)$ and $J(t)$ that we were unable to see. We were looking for something deeper, and we neglected the conceptually simple case (hard to see due to the many algebraic cancellations).  

\paragraph{Acknowledgements.}
We thank Patrick Charbonnau, Jorge Kurchan and Pierfrancesco Urbani for previous collaboration on which this work is based.

\end{document}